%% file: main.tex
\documentclass[a4paper,UKenglish, cleveref,autoref, thm-restate]{lipics-v2021}

\usepackage{amsthm, amsmath,amssymb,amsfonts}
\usepackage{cuted}
\usepackage{cite}
\usepackage[fleqn]{amsmath}
\usepackage{multicol}
\usepackage{thmtools,hyperref,cleveref}
\usepackage{enumitem}
\usepackage[font=small,labelfont=bf, figurename=Fig.]{caption}

\usepackage[dvipsnames]{xcolor}
\usepackage{xspace}
\usepackage[utf8]{inputenc}

\usepackage{algorithmic}
\usepackage{graphicx}
\usepackage{textcomp}
\usepackage{xcolor}
\usepackage{bibentry}

\usepackage{tikz}

\title{The Normalized Edit Distance with Uniform Operation Costs is a Metric}  

\titlerunning{The Normalized Edit Distance with Uniform Operation Costs is a Metric} 

\author{Dana Fisman}{Dept. of Computer Science, Ben-Gurion University, Beer-Sheva, Israel}{dana@cs.bgu.ac.il}{}{}
\author{Joshua Grogin}{Dept. of Computer Science, Ben-Gurion University, Beer-Sheva, Israel}{joshuag@post.bgu.ac.il}{}{}
\author{Oded Margalit}{Dept. of Computer Science, Ben-Gurion University, Beer-Sheva, Israel}{odedm@post.bgu.ac.il}{}{}
\author{Gera Weiss}{Dept. of Computer Science, Ben-Gurion University, Beer-Sheva, Israel}{geraw@cs.bgu.ac.il}{}{}

\authorrunning{Fisman, Grogin, Margalit, Weiss} 

\Copyright{Fisman, Grogin, Margalit, Weiss} 

\input{macros}


\nolinenumbers 

\EventEditors{Hideo Bannai and Jan Holub}
\EventNoEds{2}
\EventLongTitle{33rd Annual Symposium on Combinatorial Pattern Matching (CPM 2022)}
\EventShortTitle{CPM 2022}
\EventAcronym{CPM}
\EventYear{2022}
\EventDate{June 27--29, 2022}
\EventLocation{Prague, Czech Republic}
\EventLogo{}
\SeriesVolume{223}
\ArticleNo{14}

\begin{document}

\maketitle

\begin{abstract}
 We prove that the normalized  edit distance proposed in [Marzal and Vidal 1993] is a metric when the cost of all the edit operations are the same. This closes a long standing gap in the literature where several authors noted that this distance does not satisfy the triangle inequality in the general case, and that it was not known whether it is satisfied in the uniform case --- where all the edit costs are equal.
 We compare this metric to two normalized metrics proposed as alternatives in the literature, when people thought that Marzal's and Vidal's distance is not a metric, and identify key properties that explain why the original distance, now known to also be a metric, is better for some applications. Our examination is from a point of view of formal verification, but the properties and their significance are stated in an application agnostic way.
 \end{abstract}

\section{Introduction}
\label{sec:intro}
\input{intro_rev}

\vspace{-2mm}
\section{Preliminaries}
\label{sec:prelim}
\vspace{-1mm}
\input{defs_rev}

\section{A Proof of the Triangle Inequality}
\label{sec:triang}
\input{triang_rev}

\section{Conclusions}

We closed a gap regarding the normalized version of the editing distance proposed by Marzal and Vidal, denoted here as \ned. Marzal and Vidal noted that \ned\ is not a metric in general and left open the question of whether it is a metric in case all weights are equal. This open point, spawned two versions of a 
normalized  editing distance that have been proven to be metrics --- \ged\ and \ced. We proved that, with uniform weights, \ned\ is also a metric. To pinpoint the benefits of \ned\ over the other distances we have defined a number of properties that \ned\ maintains and \ced\ and/or \ged\ do not. The motivation for formulating the properties as we did comes from formal verification, so is our interest in uniform weights.


\bibliographystyle{plainurl}

\bibliography{bib.bib}

\appendix

\section{Appendix}\label{app:proofs}

We provide here two proofs that we could not fit in the body of the paper.

\begin{customthm}{\autoref{prop:reflexivity-and-symmetry}}[restated]
Let $s,s_1,s_2\in\Sigma^*$. Then
\begin{multicols}{2}
\begin{enumerate} 
\item $\ned(s,s)=0$ 
\item if $s_1\neq s_2$ then $\ned(s_1,s_2)>0$
\item $\ned(s_1,s_2)=\ned(s_2,s_1)$
\end{enumerate}
\end{multicols}
\end{customthm}

\begin{proof}
	 First clearly, if $s\neq\varepsilon$ then $\esd^{|s|}$ is an edit path from $s$ to $s$, and thus $\ned(s,s)=\frac{0}{|s|}=0$. 
	 Second, if $s_1\neq s_2$ then any edit path from $s_1$ to $s_2$ must contain at least one non-\esd character. Thus, its cost is $\frac{d}{l}$ for some $d>0$, implying $\ned(s_1,s_2)>0$.
	Third, assume $p_{12}=\gamma_1\gamma_2\ldots\gamma_k$ is an edit path from $s_1$ to $s_2$.
	Define $\overline{p_{12}}=\overline{\gamma_1}\, \overline{\gamma_2} \ldots\overline{\gamma_k}$ where 
	\[\overline{\gamma}=\left \{ \begin{array}{ll}
															\esd_\sigma & \mbox{if } \gamma=\esd_\sigma \\
															\esc_{(\sigma_2,\sigma_1)} & \mbox{if } \gamma=\esc_{(\sigma_1,\sigma_2)} \\
															\esx_\sigma & \mbox{if } \gamma=\esv_\sigma\\
															\esv_\sigma & \mbox{if } \gamma=\esx_\sigma																														
														\end{array}\right .\]
	Then $\overline{p_{12}}$ is an edit path from $s_2$ to $s_1$ and the cost they induce is the same.
	Hence, if  ${p_{12}}$ is a minimal edit path from $s_1$ to $s_2$ then  $\overline{p_{12}}$ is a minimal edit path from $s_2$ to $s_1$ implying $\ned(s_1,s_2)=\ned(s_2,s_1)$. 	\end{proof}

\begin{customthm}{\autoref{claim:alignment}}[restated]
Let $\Sigma,\Sigma_1,\Sigma_2$ be disjoints nonempty alphabets. 
	Let $s'_1\in\Sigma\uplus\Sigma_1$  and $s'_2\in\Sigma\uplus \Sigma_2$ 
	and $p'$ an edit path transforming $s'_1$ to $s'_2$. There exists an edit path $p$ transforming 
	$\pi_\Sigma(s'_1)$ to $\pi_\Sigma(s'_2)$ 
	such that $\cost(p) \leq \cost(p')$.
\end{customthm}

\begin{proof}
Let $\gamma\in\Gamma$, $p'\in\Gamma_{\Sigma\cup\uplus\Sigma_1\uplus\Sigma_2}^*$.
We define  $f:\Gamma_{\Sigma\uplus\Sigma_1\uplus\Sigma_2} \rightarrow \Gamma_{\Sigma}$ as follows
  \[
  f(\gamma)=\left \{ \begin{array}{ll}
     \esl_\sigma & \text{if } \gamma = \esl_\sigma \text{ for some } \esl \in \{\esv, \esx, \esd\} \text{ and } \sigma\in\Sigma\\
     c_{\sigma,\sigma'} & \text{if } \gamma = c_{\sigma,\sigma'} \text{ and } \sigma,\sigma'\in\Sigma\\
     \esv_\sigma & \text{if } \gamma = c_{\sigma_1,\sigma} \text{ and } \sigma_1\in\Sigma_1,\ \sigma\in\Sigma\\
     \esx_\sigma & \text{if } \gamma = c_{\sigma,\sigma_2} \text{ and } \sigma\in\Sigma,\ \sigma_2\in\Sigma_2\\
     \varepsilon & \text{otherwise }
  \end{array}\right.
  \]
Let $p = f(p')$  where $f:\Gamma_{\Sigma\uplus\Sigma_1\uplus\Sigma_2}^* \rightarrow \Gamma_{\Sigma}^*$ is the natural extension of $f$
defined by $f(\gamma_1\ldots\gamma_m)=f(\gamma_1)\ldots f(\gamma_n)$.

It is not hard to see that $p$ is an edit path from $\pi_\Sigma(s'_1)$ to $\pi_\Sigma(s'_2)$. Since all removed edit operations have cost $1$ we get from \autoref{plus1} that $\cost(p) \leq \cost(p')$
\end{proof}

\end{document}

%% file: macros.tex
\newcommand{\commentout}[1]{}

\renewcommand{\paragraph}[1]{\vspace{3mm}{\emph{\fontsize{11}{12}\sffamily\bfseries{#1}}}}

\newcommand{\ed}{\textsc{ed}\xspace}
\newcommand{\ned}{\textsc{ned}\xspace}
\newcommand{\ced}{\textsc{ced}\xspace}
\newcommand{\cedprime}{\textsc{ced'}\xspace}

\newcommand{\ged}{\textsc{ged}\xspace}

\newcommand{\ptime}{\textsc{Ptime}}

\newcommand{\editsymbol}[1]{{\text{\texttt{#1}}}\xspace}
\newcommand{\esl}{\editsymbol{l}\xspace}
\newcommand{\esx}{\editsymbol{x}\xspace}
\newcommand{\esv}{\editsymbol{v}\xspace}
\newcommand{\esb}{\editsymbol{b}\xspace}
\newcommand{\esd}{\editsymbol{n}\xspace}
\newcommand{\esc}{\editsymbol{c}\xspace}

\theoremstyle{property}
\newtheorem{property}[theorem]{Property}
\theoremstyle{Claim}
\newtheorem{Claim}[theorem]{Claim}

\newtheorem{innercustomthm}{} 
\newenvironment{customthm}[1]
{\renewcommand\theinnercustomthm{#1}\innercustomthm}
{\endinnercustomthm}

\newcommand{\blank}{\ensuremath{\texttt{\_}}}
\newcommand{\blanknew}{\ensuremath{\sigma_{new}}}

\newcommand{\theysaid}[1]{{\textsl{``#1''}}}

\newcommand{\weight}{\mathop{wgt}}
\newcommand{\len}{\mathop{len}}
\newcommand{\cost}{\mathop{cost}}
\newcommand{\apply}{\mathop{apply}}

\newcommand{\compose}{\mathop{cmps_h}}
\newcommand{\Compose}{\mathop{cmps}}

\newcommand{\ignore}[1]{}
\newcommand{\nobibentry}[1]{{\let\nocite\ignore\bibentry{#1}}}

%% file: intro_rev.tex
The \emph{edit distance}~\cite{Levenshtein66}, also called \emph{Levenshtein distance}, is the minimal number of insertions, deletions or substitutions of characters needed to edit one word into another. This is a commonly used measure of the distance between strings. It is used in error correction, pattern recognition, computational biology, and other fields where the data is represented by strings.  

One limitation of the edit distance is that it does not contain a normalization with respect to the lengths of the compared strings. This limits its use because, in many applications, having many edit operations when comparing short strings is more  significant than having the same number of edit operations in a comparison of longer strings, i.e., some applications require a measure that  captures the `average' number of operations per letter, in some sort. 

There are several approaches in the literature to add a normalization factor to the edit distance, as follows. The simplest idea that comes to mind is, of course, to divide the edit distance by the sum of lengths of the strings. However, Vidal and Marzal~\cite{MarzalV93} showed that this function, termed \emph{post-normalized edit distance} in~\cite{MarzalV93}, does not satisfy the triangle inequality, and thus is not a metric. Dividing by the length of the minimal or maximal among the strings also breaks the triangle inequality~\cite{HigueraM08}.
%
The fact that a distance measure is (or is not) a metric allows (resp. prevents) optimizations in many applications. For example, many efficient algorithms for searching shortest paths in graphs, such as Dijkstra's algorithm, make use of the fact that the underlying distance is a metric.

Vidal and Marzal propose thus another function, that we will focus on in this paper, that they term \emph{the normalized edit distance} (\ned) and say that this function, \theysaid{seems more likely to fulfill the triangle inequality}.
They however, show that when the sum of the costs of deleting and inserting a particular symbol is much smaller than any other elemental edit cost the function that they suggest is also non-triangular. The question of whether this distance is triangular in less contrived situations is given only an empirical answer --- \theysaid{triangular behavior has actually been observed in practice for the normalized edit distance}. 
This state of affairs opened the way for attempts to define edit distance functions that are normalized and satisfy the triangle inequality, 
as discussed in the following two paragraphs.%
\footnote{The complexity of computing \ned\ was first shown to be $O(mn^2)$ with experimental data that suggested that it is actually $O(mn)$~\cite{VidalMA95}. It was later proven to be $O(mn \log n)$ in the uniform case~\cite{arslan2000efficient}. Here, $n\geq m$ are the lengths of the compared words.}

Li and Liu~\cite{YujianB07} proposed an alternative normalization method. They open their paper by saying that \theysaid{Although a number of normalized edit distances presented so far may
offer good performance in some applications, none of them can be regarded as a genuine metric between strings because they do not satisfy the triangle inequality}. They, then, define a new distance, \emph{the generalized edit distance} (\ged), that is a simple function of the lengths of the compared strings and the edit distance between them and show that it 
is a metric.

De la Higuera and Mi\`{c}o~\cite{HigueraM08} propose the \emph{contextual normalised edit distance} (\ced).
Their normalization goes by  dividing each edit operation locally by the length of the string on which it is applied. Specifically, instead of dividing the total edit costs by the length of the edit path, they propose to divide the cost of each edit operation by the length of the string at the time of edit. They prove that this is a metric, provide an efficient approximation procedure for it, and demonstrate its performance in several application domains. 

In this paper we prove that \ned, the original edit normalization approach proposed by Vidal and Marzal~\cite{MarzalV93} does satisfy the triangle inequality when the cost of all the edit operations are the same. Since this setup is very common in many applications of the edit distance, our result gives a simple normalization technique that satisfies the triangle inequality. While there are other normalized edit distance functions that are a metric, in particular the two mentioned above (\ged and \ced), their definition is more complicated and they capture a different notion of distance than that of \ned.

The motivation that led us to engage in distances between words came from the field of formal methods; specifically, for software verification. 
In this field, it is customary to represent runs of a system using words and analyze the relationship between the set of words that satisfy 
a given specification and the set of words that the system under examination produces. Naturally, the main question asked is whether there is a word that the system produces that does not satisfy the requirement, but an appropriate concept of distance opens up the possibility of asking further questions. For example, for systems that meet the specifications, the robustness question would be, ``is there a run that is closer than a given threshold to not meeting the requirements?''. In this context, we would like the distance to measure how much ``disturbance'' in a word we can afford without risking non-compliance. Naturally, since editing model symmetric disturbances, we use uniform weights. As we will explain in \autoref{sec:comparison} below, 
 the \ned distance satisfies certain properties required for use in formal the field of formal methods
 that other metrics do not. Another advantage of \ned in the context of formal methods is that its definition allows direct use of a \ptime\ algorithm 
proposed by Filliot et al.~\cite{FiliotMR0020}
for computing the distance between regular sets of words represented using finite automata. 
This is useful since verification tools work with automata to represent the specification and the program runs, and verification questions are usually reduced to questions on automata.

%% file: defs_rev.tex
Let $\Sigma$ be a finite alphabet and $\Sigma^*$ the set of all finite strings over $\Sigma$. The length of string $w=\sigma_1\sigma_2\ldots\sigma_n$, denoted $|w|$, is $n$.
We use $w[i]$ for the $i$-th letter of $w$, and $w[i..]$ for the suffix of $w$ starting at $i$, namely $w[i..]=\sigma_i\sigma_{i+1}\ldots\sigma_n$.
 
\paragraph{Basic and extended edit letters}
The literature on defining distance between words over $\Sigma$ uses the notion of \emph{edit paths}, which are strings over \emph{edit letters} defining how to transform a given string $s_1$ to another string $s_2$.
The standard operations are \emph{deleting} a letter, \emph{inserting} a letter, or \emph{swapping} one letter with another letter. Formally, the \emph{basic edit letters alphabet} $\Gamma$ is defined as 
$\Gamma=\{\esd, \esc, \esv, \esx\}$ where:
\vspace{-2mm}
\begin{multicols}{2}
\begin{itemize}
	\item \esc stands for \emph{change}: 
	the relevant letter in the source string is replaced with another letter.  
	\item \esv stands for \emph{insert}: 
	a new letter is added to the destination string.  \\
	\item \esx stands for \emph{delete}: 
	the current letter from the source string is deleted and not copied to the destination string.  
	\item \esd stands for  \emph{no-change}: 
	 the current letter is copied as is from the source string to the destination string.  
\end{itemize}
\end{multicols}
\vspace{-3mm}
The edit letters in $\Gamma$ do not carry enough information to transform a string $w$ over $\Sigma$ to an unknown string  over $\Sigma$,
since for instance the letter $\esv$ does not provide information on which letter $\sigma\in\Sigma$ should be inserted. To this aim we define
the alphabet   
$\Gamma_\Sigma$ that provides all the information required. Formally, 
$\Gamma_\Sigma = \{ \texttt{l}_{\sigma}|~\sigma\in\Sigma, \texttt{l}\in\{\esd,\esv,\esx\}\} \cup
\{ \esc_{(\sigma_1,\sigma_2)}|~\sigma_1,\sigma_2\in\Sigma\}$.
We call strings over $\Gamma_\Sigma$ \emph{edit paths}. 
Throughout this document we use $w,w_1,w_2,w',\ldots$ and $s,s_1,s_2,s',\ldots$ for strings over $\Sigma$ and 
$p,p_1,p_2,p',\ldots$ for edit paths.

\paragraph{Weights and length of edit paths}\label{par:weights}
Given a function $\weight\colon\Gamma_\Sigma\rightarrow \mathbb{N}$, that defines  a weight to each
edit letter, we define the weight of an edit path  $\weight\colon\Gamma_{\Sigma}^*\rightarrow \mathbb{N}$ 
as the sum of weights of the letter it is composed from,
namely for an edit path $p=\gamma_1\gamma_2\ldots\gamma_m\in\Gamma^*_\Sigma$, $\weight(\gamma_1\ldots\gamma_m) =\sum_{i=1}^{m}\weight(\gamma_i)$.

In our case we are interested in \emph{uniform} costs where the weight of
$\esd$ is $0$ and the weight of all other operations is the same. For simplicity
we can assume that the weight of all other operations is $1$.
Thus, we can define the weight over $\Gamma$ instead of $\Gamma_\Sigma$ simply as  
$\weight\colon\Gamma\rightarrow \mathbb{N}$ where $\weight(\gamma)=0$ if $\gamma = \esd$ and $\weight(\gamma)=1$ otherwise, namely if $\gamma\in\{\esc, \esv, \esx \}$.
\commentout{
\[\weight(\gamma)=\begin{cases} 
0 & \text{if } \gamma = \esd\\
1 &  \text{if } \gamma \in\{\esc, \esv, \esx \}\\
\end{cases}\]}
We also define the function  $\len\colon\Gamma_\Sigma\rightarrow \mathbb{N}$ as $\len(\gamma)=1$
and  $\len\colon\Gamma_\Sigma^*\rightarrow \mathbb{N}$ as $\len(\gamma_1\ldots\gamma_m) =\sum_{i=1}^{m}\len(\gamma_i)$.
Clearly here we have $\len(p)=|p|$. Later on we will introduce new edit letters whose length is different from $1$, thus the need for a definition of $\len$ that is not just the count of letters.

\begin{example}\label{example:edit-path}
Let $w_1=abcd$ and $w_2=badee$. Then $p=\esx_{a}\cdot \esd_b \cdot \esc_{c,a}\cdot\esd_d\cdot\esv_e\cdot\esv_e$ is an edit path transforming $w_1$ to $w_2$.
We have that $\weight(p)=\weight(\esx\esd\esc\esd\esv\esv)=4$ and $\len(p)=6$.
\end{example}
\vspace{-3mm}
\paragraph{Applying an edit path to a string}
Given a string $w$ over $\Sigma$, and
an edit path $p$ over $\Gamma_\Sigma$ we can now define the result
of applying $p$ to $w$.
\begin{definition}
We define a function $\apply \colon \Sigma^*\times \Gamma_\Sigma^* \rightarrow (\Sigma \cup \{\bot\})^*$ that 
given a string $w$ over $\Sigma$, and
an edit path $p$ over $\Gamma_\Sigma$  returns a new string $w'$ over $\Sigma \cup \{\bot\}$. If $p$ is a valid edit path for $w$ it  returns a string over $\Sigma$, otherwise a string that contains $\bot$.
\textup{
\[
\apply(p, w)=
\left\{ 
\begin{array}{ll}
\varepsilon & \text{if } p=w=\varepsilon \\ 
\sigma' \cdot \apply(p[2..], w) & \text{if } p[1]{=}\esv_{\sigma'}\\
\sigma' \cdot \apply(p[2..],w[2..])  & \text{if } p[1]{=}\esc_{(\sigma,\sigma')} \text{ and } w[1]{=}\sigma\\
\sigma\cdot \apply(p[2..],w[2..]) & \text{if } p[1]{=}\esd_\sigma   \text{ and } w[1]{=}\sigma\\
\apply(p[2..],w[2..])  & \text{if } p[1]{=}\esx_\sigma   \text{ and } w[1]{=}\sigma\\
\bot & \text{otherwise}
\end{array}
\right.\]
}
\end{definition}

We say that a string $p_{ij}$ over $\Gamma_\Sigma$ is an \emph{edit path from} string $s_i$ to string $s_j$ over $\Sigma$ if $\apply(p_{ij},s_i)=s_j$. With a bit of overriding, we say that a string $p_{ij}$ over $\Gamma$ is an \emph{edit path from} strings $s_i$ to $s_j$ over $\Sigma$ if there exists an extension of $p_{ij}$ with subscripts from $\Sigma$ that results in an edit path  from $s_i$ to $s_j$.
  
\begin{example}\label{example:apply}
	Following on \autoref{example:edit-path}, we have that $\apply(\esx_{a} \esd_b  \esc_{c,a} \esd_d \esv_e \esv_e,abcd)=badee$,
	and that $\esx\esd\esc\esd\esv\esv$ is an edit path from $abcd$ to $badee$.
\end{example}  

\paragraph{The normalized edit distance} 
Let $p$ be an edit path. The \emph{cost} of $p$, denoted $\cost(p)$
is defined to be the weight of $p$ divided by the length of $p$, if the length is not zero,
and zero otherwise. That is, $\cost(p)=0$ if $|p|=0$ and $\cost(p)=\frac{\weight(p)}{\len(p)}$ otherwise.
 
Using the definition of $\cost$ we can define the notion we study in this paper, namely
the \emph{normalized edit distance}, \ned, of Marzal and Vidal~\cite{MarzalV93}.  
\begin{definition}[The normalized edit distance, \ned~\cite{MarzalV93}]\label{def:ned}
  The \emph{normalized edit distance} between $s_i$ and $s_j$, denoted $\ned(s_i,s_j)$ is 
   the minimal cost of an edit path $p_{ij}$ from $s_i$ to $s_j$.
  That is, 
  \[\ned(s_i,s_j)=\min\left\{ \left. \cost(p_{ij}) ~\right|~ p_{ij}\in\Gamma_\Sigma^* \text{ and } \apply(p_{ij},s_i)=s_j\right\}\]
\end{definition}
Note that while, in general, $\weight$ may assign arbitrary weights to edit letters, in this paper we assume the uniform weights as defined above.

\begin{example}\label{example:three}
	Let $\Sigma=\{a,b,c\}$, $s_1=acbb$ and $s_2=cc$. Then the string $\esx\esd\esx\esc$ denotes an edit path taking $s_1$, deleting the first letter ($a$), copying the second letter ($c$), deleting the third letter ($b$), and replacing the fourth letter ($b$) by $c$. This edit path indeed transforms $s_1$ to $s_2$. Its cost is $\frac{1+0+1+1}{4}=\frac{3}{4}$. It is not hard to verify that this cost is minimal, therefore $\ned(s_1,s_2)=\frac{3}{4}$.
\end{example}

\paragraph{The alignment view}
Recall that distance functions defined by dividing the weight by the sum, max or min of the given strings does not yield a metric\cite{MarzalV93,HigueraM08}.
The main contribution of the paper is to show that the choice  to use the length of the edit path in the denominator, makes the resulting definition, \ned, a metric.
To understand the motivation behind dividing by the length of the edit path, note that an edit path can be thought of as defining an alignment between the given words $s_1$ and $s_2$ 
by padding the first string with some blank symbol, denote it $\blank$, whenever an insert operation is conducted, and padding the second string with $\blank$ symbols whenever
 a delete operation is conducted. The resulting words $s'_1$ and $s'_2$ would thus be of the same length, and the weight of the edit path would correspond to the Hamming distance between the words. (The Hamming distance applies only to words of same length and counts the number of positions $i$ in which the two words differ.) When dealing with words of the same length it makes sense to normalize them by dividing by their length, and the length of the padded words equals the length of the edit paths.

\begin{example}
	In \autoref{example:three} we used $s_1=acbb$, $s_2=cc$. The edit path $\esx\esd\esx\esc$ corresponds to the alignment $s'_1=acbb$ and $s'_2=\blank c \blank c$,
	and since the length of $s'_1$ and $s'_2$ is $4$ and they differ in all positions but one the corresponding cost is $3/4$.  
	 
	In \autoref{example:edit-path}, we used $w_1=abcd$ and $w_2=badee$ and considered the edit path $\esx\esd\esc\esd\esv\esv$.
	This path correspond to the alignment $w'_1= abcd\blank\blank$ and $w'_2=\blank b a d e e$. 	
	Since $w'_1$ and $w'_2$ differ in four out of the six positions, we have that the cost of this path is $4/6$.
\end{example}

	\paragraph{A metric space}
A metric space is an ordered pair $(\mathbb{M},d)$ where $\mathbb{M}$ is 
a set and $d\colon\mathbb{M}\times\mathbb{M}\rightarrow\mathbb{R}$ is a \emph{metric},
i.e., it satisfies the following for all $m_1,m_2,m_3\in\mathbb{M}$:
\vspace{-4mm}
\begin{multicols}{2}
\begin{enumerate}
	\item $d(m_1,m_2)=0$ iff $m_1=m_2$;
	\item $d(m_1,m_2)=d(m_2,m_1)$;
	\item $d(m_1,m_3)\leq d(m_1,m_2)+d(m_2,m_3).$
\end{enumerate}
\end{multicols}
\noindent
The first condition is referred to as \emph{identity of indiscernibles}, the second as \emph{symmetry}, the third as the \emph{triangle inequality}.

\paragraph{Basic properties of NED}
It is not hard to see that \ned\ satisfies the first and second condition of being a metric.
The following proposition establishes that the distance of a string to itself, according to \ned, is zero, and that the distance between two strings is symmetric.
\begin{proposition}\label{prop:reflexivity-and-symmetry}
Let $s,s_1,s_2\in\Sigma^*$. Then
\vspace{-4mm}
\begin{multicols}{2}
\begin{enumerate} 
\item $\ned(s,s)=0$ 
\item if $s_1\neq s_2$ then $\ned(s_1,s_2)>0$
\item $\ned(s_1,s_2)=\ned(s_2,s_1)$
\end{enumerate}
\end{multicols}
\end{proposition}
Its straight forward proof can be found in the archived version~\cite{NEDarxiv}.

The challenge is proving that \ned satisfies the third condition, the triangle inequality. 
We do this in \autoref{sec:triang}. Before that we investigate some properties of \ned\ and other 
edit distance functions.

\section{Properties of the various normalized edit distance functions} 
 
\subsection{Other edit distance functions}
\label{sec:other-edit-distances}
In the introduction we mentioned several edit distance functions known to be a metric.
We use the term \emph{edit distance}  for functions  between words to values that are based on \emph{delete}, \emph{insert} and \emph{swaps}.

In general these definition may allow arbitrary weight assignment to edit letters, but we consider the case of uniform weights.
We start by introducing the edit distance functions, \ed, \ged, and \ced, and then turn to compare their properties, with those of \ned.

We start with the commonly used \emph{edit distance}, introduced by Levenstein~\cite{Levenshtein66}. 
\begin{definition}[The  edit (Levenstein) distance, \ed]
   The \emph{ edit distance} between $s_i$ and $s_j$, denoted $\ed(s_i,s_j)$, is 
   the minimal  weight of a path $p_{ij}$
  from $s_i$ to $s_j$.
  That is, 
  \[\ed(s_i,s_j)=\min \left\{  \weight(p_{ij}) ~|~  p_{ij}\in\Gamma_\Sigma^* \text{ and } \apply(p_{ij},s_i)=s_j
                        \right\}\]
\end{definition}
This function is a metric, but  it completely ignores
the lengths of the words, thus it is not {normalized}.

We turn to introduce the \emph{generalized normalized edit distance} proposed and proven to be a metric by Li and Liu~\cite{YujianB07}. 
\begin{definition}[The  generalized edit distance]
$\ged(s_i,s_j)= \frac{2\cdot \ed(s_i,s_j)}{|s_i|+|s_j|+\ed(s_i,s_j)}$.
\end{definition}

Last, we define of the \emph{contextual edit distance}, proposed and proven to be a metric by de la Higuera and
               Mic{\'{o}}~\cite{HigueraM08}. 
It starts with a definition of distance between two strings whose Levenstein distance is $1$,
from which it builds the distance for an arbitrary set of words, by looking at a sequence of intermediate
transformations.
\begin{definition}[The  contextual edit distance]\label{def:ced}
Let $s,s'$ be such that $\ed(s,s')=1$ their contextual edit distance is defined by
$\ced(s,s')= \frac{1}{\max(|s|,|s'|)}$.
Note that given $\ed(s,s')=1$ the difference between the lengths of $s$ and $s'$ is at most one, thus $\max(|s|,|s'|)\leq \min(|s|,|s'|)+1$.
                    
Given a sequence of strings $\alpha=(s_0,s_1,\ldots,s_k)$ such that 
$ \ed(s_i,s_{i+1})=1$ for all $0\leq i<k$, one can define 
$\ced(\alpha)=\sum_{i=1}^k \ced(s_{i-1},s_i).$
To define the contextual edit distance between arbitrary strings $s_x$ and $s_y$
one considers the minimum of $\ced(\alpha)$ among all sequence of strings $\alpha=s_0,s_1,\ldots,s_k$ as above such that $s_0=s_x$, $s_k=s_y$.
That is, 
$\ced(s_x,s_y)=\min\left\{ \ced(\alpha)~\left|~\begin{array}{l} 
                                        \alpha=(s_0,s_1,\ldots,s_k), \  
                                        s_0=s_x,\ s_k=s_y, \ 
                                        \ed(s_i,s_{i+1})=1 
                                        \end{array}\right.\right\}$.
\end{definition}

\subsection{Comparison to other edit distance functions}\label{sec:compare}
\label{sec:comparison}

Comparing \ned\ and \ed\ is easy.  The \ned\ distance (like \ced\ and \ged) measures the average number edits, not just the total count.
To see why this is needed, consider two short words $x_1,x_2$ that differ in $k$ letters and two long word $y_1,y_2$ that also differ in $k$
letters. In the context of software verification, for example, the latter represent runs that are more similar to one another than the former. We thus, expect the distance between the $y$s to be less than the distances between the $x$s 
but this is not the case in \ed, as can be observed by inspecting the following words.
\vspace{-2mm}
$$\begin{array}{l@{\qquad}l}
\ed(aabcde,abpcg)=4 & \ned(aabcde,abpcg)={4}/{7}\\ 
\ed(a^{96}b^4,a^{100})=4 & \ned(a^{96}b^4,a^{100})={4}/{100}\\ 
\end{array}$$

We turn to a comparisons of $\ned$ with the other normalized edit distances, \ged\ and \ced. Usually, being normalized means that the values of the distance functions are bounded within a given range, but this is not always the case.
The lower bound is clearly $0$ for \ned, \ged, and \ced, since they are metric. The upper value of \ned\ and \ged\ is $1$ but the values for \ced\ are not bounded:

\begin{Claim}
	The values of \ned and \ged cannot exceed $1$ and may reach $1$, the values of \ced are unbounded.
\end{Claim}	
\begin{proof}
	For \ned\ the numerator is the weight of an edit path, which is always smaller than the denominator which is the length of the edit path, thus $\ned(w_1,w_2)\leq 1$ for all $w_1,w_2\in\Sigma^*$. Since $\ned(\varepsilon,a)=1$ the upper bound is $1$.
	
	For \ged\ the numerator is twice the weight of the edit path, and the denominator is once the weight of the edit path, plus the sum of length of the strings which is at least the size of the edit path, thus clearly at least the weight of the edit path. This shows \ged\ cannot exceed $1$. The fact that $\ged(\varepsilon,a)=1$ shows that $1$ is the upper bound.
	
	To see why \ced\ is not bounded consider the sequence of words $\{a^i\}_{i\in\mathbb{N}}$. That is, the sequence $\varepsilon,a,aa,aaa,\ldots$.
	We have that $\ced(\varepsilon,a^i)=1+\frac{1}{2}+\frac{1}{3}+\ldots+\frac{1}{i}$. 
	Thus $\ced(\varepsilon,a^i)$ is the sum of the Harmonic sequence up to the $i$th element, and since the Harmonic sequence diverges, $\ced$ is unbounded. 
\end{proof}

Towards the second property of metrics that we consider, recall that the first requirements of a metric, \emph{identity of indiscernibles}, is that $d(s_1,s_2)=0$ if and only if $s_1=s_2$.
That is, the distance between two strings (in our case) is zero if and only if it is the exact same string.
In the case of strings, when working with a normalized distance with an upper bound $1$, 
we expect the distance to be $1$,
the maximal possible, if the strings are completely different, namely they do not have any letter in common,
that is, for all $\sigma\in\Sigma$ if $\sigma$ appears in $s_1$ it does not appear in $s_2$ and vice versa. In software verification, for example, this means that the system produced a run that is completely unrelated to the specification, thus we expect the distance to be $1$, indicating it is as far away as possible from the specification.

Since $\ced$ is unbounded, we consider  for the purpose of the next property, a slightly different version, that we call \cedprime, defined as $\cedprime(s_1,s_2)=\min(1,\ced(s_1,s_2))$.
\footnote{This is inspired by \cite{CEDwebsite} 
 that explains this choice as follows:
\theysaid{This measure is not normalized to a particular range. Indeed, for a string of infinite length and a string of 0 length, the contextual normalized edit distance 
would be infinity. But so long as the relative difference in string lengths is not too 
great, the distance will generally remain below 1.0}.}

\begin{property}[max variance of antitheticals]
  Let $d\colon \Sigma^*\times\Sigma^* \rightarrow [0,1]$ be an edit distance function.
  We say that $d$ has the property of \emph{max variance of antitheticals} if
  $d(s_1,s_2)=1$ if and only if $s_1$ and $s_2$ have no letter in common. 
\end{property}
We show that $\ned$ has this property  while $\ged$ and \cedprime do not.\footnote{Note that extending this property to require that $d(s_1,s_2)$ equals the maximal value (be it $1$ or more) only for antitheticals, so
that it can be applied to the original \ced, would not make \ced satisfy it since $\ced(\varepsilon,a) = 1 < \infty$.}
\begin{Claim}
 The property of {max variance of antitheticals} holds for \ned, but does not hold for \ged\ and \cedprime. 
\end{Claim}
\begin{proof}
    Consider $aa$ and $bb$. Since they have no common letter, we expect their distance to be $1$.
    The  fact that $\ged(aa,bb)={2}/{3}$ shows that \ged\ violates the property of {max variance of antitheticals}.\footnote{We note that, moreover, $\ged(aab,b)$ is also ${2}/{3}$ though we expect $\ged(aab,b)<\ged(aa,bb)$ since the average number of edits is smaller in the first case.}
    Consider $a$ and $aaaa$. Since they do have a common letter, we expect their distance to be strictly less than $1$.
    The fact that $\ced'(a,aaaa)=1$ shows that \cedprime\ violates the property of {max variance of antitheticals}.
    
    To see that \ned\ has this property, note that it results in a value of $1$ iff the numerator equals the denominator,
    i.e., the weight of the edit path is the same as its length; which holds iff there are no edit letters with weight zero.
    Since the only zero weight edit letter is no-change, $\esd$, the value of \ned\ is $1$ if and only if the words have no common letter.
\end{proof}

For the third metric comparison property, consider two words $u$ and $v$ and suppose $d(u,v)=c$ for the concerned edit distance function $d$.
When considering normalized edit distance, we expect that $d(u^i,v^i)$ will not exceed $c$ since by repeating $i$ times the  edit operations for transforming $u$ into $v$ we should be able to transform $u^i$ into $v^i$ and the `average' number of edits will not change. It could be that when considering the longer words $u^i$ and $v^i$ there is a  better sequence of edits, thus we do not expect equality. As before, our motivation for requiring this property comes from software verification. Specifically, when considering {periodic runs}, generated, e.g., by code with loops, one would expect that the distance between the periodic runs is not larger than the distance between the periods because an error that repeats regularly should only be counted once in a normalized measure that models average error rate.

\begin{property}[Non escalation of repetitions]
    Let $d$
    be an edit distance function.
    Let $u,v\in\Sigma^*$.
    If $d(u^k,v^k)\leq d(u,v)$ for any $k>1$ we say that $d$ \emph{does not escalate repetitions}. 
\end{property}
\begin{Claim}
    The \ned\ and \ged\ distances satisfy the property of non escalation of repetitions. The \ced\ and \cedprime\ distances do not.
\end{Claim}

\begin{proof}
Consider $u=aab$ and $v=aaab$. The following shows that \ced\ and \cedprime\ escalate repetitions.
\vspace{-1mm}
$$\begin{array}{l@{\qquad}l}
\ced((aab)^1,(aaab)^1)=\frac{1}{4}=0.25 \\
\ced((aab)^2,(aaab)^2)=\frac{1}{7}+\frac{1}{8}=\frac{15}{56}=0.2678 \\ 
\ced((aab)^3,(aaab)^3)=\frac{1}{10}+\frac{1}{11}+\frac{1}{12}=\frac{181}{660}=0.2742 
\end{array}$$

To see that \ned\ does not escalate repetitions, 
assume $p_{uv}$ is an optimal edit path transforming $u$ to $v$. Since $(p_{uv})^k$, the edit path obtained by repeating $k$ times $p_{uv}$,
is an edit path transforming $u^k$ to $v^k$:
$$\begin{array}{l}
    \ned(u^k,v^k)\leq \frac{k \cdot \weight(p_{uv})}{k \cdot \len(p_{uv})}=\frac{\weight(p_{uv})}{\len(p_{uv})}=\ned(u,v).
\end{array}$$
The same reasoning shows that \ged\ does not escalate repetitions. 
$$\begin{array}{l}
    \ged(u^k,v^k)\leq \frac{2k\cdot\ed(u,v)}{k(|u|+|v|)+k\cdot\ed(u,v)}=\frac{2\cdot\ed(u,v)}{|u|+|v|+\ed(u,v)}=\ged(u,v).
\end{array}\phantom{----------}\qedhere $$ 
\end{proof}

The last property we consider is referred to as \emph{pure uniformity of operations}.
While we assume the weights of delete, insert and substitution are uniform,
the resulting edit distance function may not be purely uniform, in the following 
sense. Consider two strings $s_1$ and $s_2$ such that $s_1$ is shorter than
$s_2$. Then to transform $s_1$ to $s_2$ we would need some insertion operations.
Consider now a word $s'_1$ that is longer than $s_1$ but not longer than $s_2$
and is obtained by padding $s_1$ with some new letter $\blanknew$ in some
arbitrary set of positions. Since insert and substitution weigh the same,
we expect $d(s_1,s_2)$ to be equal to $d(s'_1,s_2)$.

To define this formally we use the following notations.
Let $\Sigma'\subseteq\Sigma$ and $s\in\Sigma^*$
we use $\pi_{\Sigma'}(s)$ for the string obtained
from $s$ by leaving only letters in $\Sigma'$.
For instance, if $\Sigma=\{a,b,c\}$ and
$s=abcbacc$ then $\pi_{\{a,b\}}=abba$.

\begin{property}[pure uniformity]\label{prop:pure-uniformity}
\textup{
Let $\Sigma,\Sigma_1,\Sigma_2$ be disjoints alphabets, and
let $s_1,s_2\in\Sigma^*$.
We call $d$ \emph{purely uniform} if 
 $d(s_1,s_2)= \min \left\{ d(s'_1, s'_2) ~\left|~ 
                    s'_i\in(\Sigma\uplus\Sigma_i)^*  \text{ and }
                    \pi_\Sigma(s'_i){=} s_i 
                    \text{ for } i\in\{1,2\}
                 \right.\right\}.$
}
\end{property}

We can now show that \ned\ satisfies this property while \ged\ and \ced\ do not.

\begin{Claim}
The \ned\ distance is purely uniform. The \ged\ and \ced\ distances are not.
\end{Claim}
\begin{proof}
To see why \ged and \ced are not purely uniform consider the words $s_1=a^{50}$, $s_2=a^{100}$
and $s'_1=a^{50}c^{50}$ and note that $\pi_{\{a,b\}}(s'_1)=s_1$.
We have that $\ged(a^{50},a^{100}) = 2\cdot 50/(150+50) = 1/2$ whereas $\ged(a^{50}c^{50}    ,a^{100}) = 100/(200+100) =1/3$.
Considering \ced, we have that $\ced(a^{50},a^{100}) = \sum_{i=51}^{100} \frac{1}{i} \approx 0.68817$ whereas $\ced(a^{50}c^{50}    ,a^{100}) = \sum_{i=51}^{100} \frac{1}{100} = 0.5$.
Since all
values are below $1$, the
same is true for $\cedprime$.

To show that \ned\ is purely uniform we first note that 
$s_1,s_2\in \Sigma^*$ implies $s_1,s_2$ are in $(\Sigma\uplus\Sigma_1)^*$ and $(\Sigma\uplus\Sigma_2)^*$, respectively, 
thus the $\geq$ direction of the equality in \autoref{prop:pure-uniformity} clearly holds. 
For the $\leq$ direction, we turn to \autoref{claim:alignment} below,
which essentially formalized the intuition provided regarding the \emph{alignment view} of \ned.
Thus, given $s'_1$ and $s'_2$ establishing the min in the RHS of \autoref{prop:pure-uniformity}, and  $p'\in\Gamma^*$ an edit path transforming $s'_1$ into $s'_2$,
we can build an edit path $p\in\Gamma^*$ transforming $\pi_\Sigma(s'_1)$ into $\pi_\Sigma(s'_2)$ such that $\cost(p)\leq\cost(p')$.
This shows that $\ned(s_1,s_2)\leq \ned(s'_1,s'_2)$ for every such $s'_1,s'_2$.
Thus \ned\ satisfies the pure uniformity property.
\end{proof}

\begin{Claim}\label{claim:alignment} 
Let $\Sigma,\Sigma_1,\Sigma_2$ be disjoints nonempty alphabets. 
	Let $s'_1\in\Sigma\uplus\Sigma_1$  and $s'_2\in\Sigma\uplus \Sigma_2$
	and $p'$ an edit path transforming $s'_1$ to $s'_2$. There exists an edit path $p$ transforming 
	$\pi_\Sigma(s'_1)$ to $\pi_\Sigma(s'_2)$ 
	such that $\cost(p) \leq \cost(p')$.
\end{Claim}

%% file: triang_rev.tex
This section is the main contribution of the paper --- showing that $\ned$ with uniform costs satisfies the triangle inequality. 

Let $s_1,s_2,s_3\in\Sigma^*$ and $p_{12},p_{23}$ be edit paths, such that $\apply(p_{12},s_1)=s_2$, $\apply(p_{23},s_2)=s_3$. 
We would like to define a method $\Compose \colon \Gamma_\Sigma^*\times\Gamma_\Sigma^*\rightarrow \Gamma_\Sigma^*$
that given the two edit paths
$p_{12},p_{23}$ returns an edit path $p_{13}$ from $s_1$ to $s_3$.
In addition, using the notations $d_{*}=\weight(p_{*})$ and $l_{*}=\len(p_{*})$ for $*\in\{12,23,13\}$,
we would like to show that both of the following hold:
\vspace{-10mm}
\begin{multicols}{3}
\begin{equation}
     d_{13} \le d_{12}+d_{23} \label{d} 
\end{equation}
\begin{equation*}
\end{equation*}
\begin{equation}
	l_{13} \ge \max\{l_{12},l_{23}\} \label{l} 
\end{equation}
\end{multicols}
	
From these two equations we can deduce that the cost of the resulting path $p_{13}$ is at most  the sum of costs  of the given paths $p_{12}$ and $p_{23}$ 
proving that $\ned$ satisfies the triangle inequality.

\paragraph{Introducing a new edit letter}
To do this we need, for technical reasons, to introduce a new edit letter, which we denote $\esb$ (for \emph{blank}). This is actually an
abbreviation of $\esv\esx$, that is, it signifies that a new letter is added and immediately deleted.
We enhance the weight and length definition from $\Gamma$ to $\Gamma\cup\{\esb\}$ as follows.
\[
\begin{array}{l@{\qquad\qquad}r}
    \weight(\gamma)=\begin{cases} 
    0 & \text{if } \gamma = \esd\\
    1 &  \text{if } \gamma \in\{\esc, \esv, \esx \}\\
    2 & \text{if } \gamma = \esb\\
    \end{cases}
    &
    \len(\gamma)=\begin{cases} 
    1 & \text{if } \gamma \in\{\esd, \esc, \esv, \esx \}\\
    2 & \text{if } \gamma = \esb\\
    \end{cases}
\end{array}
\]

As before we use the natural extensions of $\weight$ and $\len$ from letters to strings and define $\cost(p)$ to be $\weight(p)/\len(p)$.

\paragraph{The compose method}
We define a helper function $\compose$  that
produces a string over $(\Gamma_\Sigma \cup \{\esb\})^*$
(rather than over $\Gamma_\Sigma^*)$. Given such a sequence we can convert it into a sequence over $\Gamma_\Sigma$
by deleting all $\esb$ symbols.
The method 
$\compose\colon \Gamma_\Sigma^*\times\Gamma_\Sigma^*\rightarrow (\Gamma_\Sigma \cup \{\esb\})^* \cup \{\bot\}$
is defined
inductively, in \autoref{def:compose}, by scanning the letters of the given edit paths $p_{12},p_{23}$.
We say that $\compose$ is well defined if it does not return $\bot$.
We show that, when applied on edit paths $p_{12}$ and $p_{23}$ transforming some $s_1$ into $s_2$ and $s_2$ into $s_3$,
respectively, $\compose$ is well defined.

\begin{definition}\label{def:compose}
Let $p_{12},p_{23}$ be edit paths over $\Gamma_\Sigma$. We define $\compose(p_{12}, p_{23})$
inductively as follows.\\
\scalebox{.9}{
\textup{
\[
\compose(p_{12}, p_{23})=
            \left\{
            \begin{array}{lll}
			\varepsilon  & \text{if } p_{12}=p_{23}=\varepsilon & (0) \\
			\esx_\sigma \cdot \compose(p_{12}[2..],p_{23})  & \text{if } p_{12}[1]=\esx_\sigma & (1) \\
			\esv_\sigma\cdot \compose(p_{12},p_{23}[2..])  & \text{if } p_{23}[1]=\esv_\sigma & (2) \\
			\esd_\sigma\cdot \compose(p_{12}[2..],p_{23}[2..])  & \text{if } (p_{12}[1],p_{23}[1]) = (\esd_\sigma,\esd_\sigma) & (3)\\
			\esc_{(\sigma',\sigma)}\cdot \compose(p_{12}[2..],p_{23}[2..])  & \text{if } (p_{12}[1],p_{23}[1]) = (\esd_{\sigma'
			},\esc_{(\sigma',\sigma)}) & (4)\\
			\esx_\sigma\cdot \compose(p_{12}[2..],p_{23}[2..])  & \text{if } (p_{12}[1],p_{23}[1]) = (\esd_\sigma,\esx_\sigma) & (5) \\
			\esc_{(\sigma_1,\sigma_3)}\cdot \compose(p_{12}[2..],p_{23}[2..])  & \text{if } (p_{12}[1],p_{23}[1]) = (\esc_{(\sigma_1,\sigma_2)},\esc_{(\sigma_2, \sigma_3)}) & (6)\\
			\esx_{\sigma_1}\cdot \compose(p_{12}[2..],p_{23}[2..])  & \text{if } (p_{12}[1],p_{23}[1]) = (\esc_{(\sigma_1,\sigma_2)},\esx_{\sigma_2}) & (7) \\
			\esc_{(\sigma',\sigma)}\cdot \compose(p_{12}[2..],p_{23}[2..])  & \text{if } (p_{12}[1],p_{23}[1]) = (\esc_{(\sigma',\sigma)},\esd_\sigma) & (8)\\
			\esv_\sigma\cdot \compose(p_{12}[2..],p_{23}[2..]) & \text{if } (p_{12}[1],p_{23}[1]) = (\esv_\sigma,\esd_\sigma)  & (9)\\
			\esv_{\sigma_2}\cdot \compose(p_{12}[2..],p_{23}[2..])  & \text{if } (p_{12}[1],p_{23}[1]) = (\esv_{\sigma_1},\esc_{(\sigma_1,\sigma_2)}) & (10)\\
			\esb\cdot \compose(p_{12}[2..],p_{23}[2..])  & \text{if } (p_{12}[1],p_{23}[1]) = (\esv_\sigma,\esx_\sigma)
			& (11)\\
			\bot & \text{otherwise} & (12)
		 \end{array}
		 \right.
\]
}}
\end{definition}

\begin{figure}[t]
\noindent\makebox{
\scalebox{.8}
{
\begin{subfigure}{.5\textwidth}

  \scalebox{.7}{\input{tikz/x_y}}
  \caption{An optimal edit path $p_{12}$ for $w_1,w_2$}
  \label{fig:p12}
\end{subfigure}
}}
\noindent\makebox{
\scalebox{.8}
{
\begin{subfigure}{.7\textwidth}

  \scalebox{.7}{\input{tikz/x_z_cmps}}
  \caption{The composed edit path $p_{13}$ using \autoref{def:compose}}
  \label{fig:p13}
\end{subfigure}
}}
\noindent\makebox{
\scalebox{.8}
{
\begin{subfigure}{.5\textwidth}

  \scalebox{.7}{\input{tikz/y_z}}
  \caption{An optimal edit path $p_{23}$ for $w_2,w_3$}
  \label{fig:p23}
\end{subfigure}
}}
\noindent\makebox{
\scalebox{.82}
{
\begin{subfigure}{.5\textwidth}
    \begin{tabular}{ll}
     & $p_{12} = \esc \esv \esd \esv \esd \esd$\\
     & $p_{23} = \esv \esd \esc \esd \esx \esd \esd$\\
     & $\compose(p_{12},p_{23}) =$\\
    & \[
    \begin{array}{llll@{\ \ }l}
        =   & \compose(\esc \esv \esd \esv \esd \esd,\esv \esd \esc \esd \esx \esd \esd) \\
        =_1 & \esv\cdot \compose(\esc \esv \esd \esv \esd \esd, \esd \esc \esd \esx \esd \esd) & \text{case (2)} & \varepsilon & \esv_a \\
        =_2 & \esv \esc\cdot \compose(\esv \esd \esv \esd \esd,\esc \esd \esx \esd \esd) & \text{case (8)} & \esc_{b,a}  &\esd_b  \\
        =_3 & \esv \esc \esv  \cdot \compose(\esd \esv \esd \esd,\esd \esx \esd \esd) & \text{case (10)} & \esv_{c} & \esc_{c,a}  \\ 
        =_4 & \esv \esc \esv \esd \cdot \compose(\esv \esd \esd,\esx \esd \esd) & \text{case (3)} & \esd_{b} & \esd_{b}  \\ 
        =_- & \esv \esc \esv \esd \esb \cdot \compose(\esd \esd,\esd \esd) & \text{case (11)} & \esv_{b} & \esx_{b}  \\ 
        =_5 & \esv \esc \esv \esd \esb \esd\cdot \compose(\esd,\esd)  & \text{case (3)} & \esd_{a} &\esd_{a}  \\ 
        =_6 & \esv \esc \esv \esd \esb \esd \esd & \text{case (3)} & \esd_{b} & \esd_{b}  \\ [1mm]
    \end{array}
    \] \\
    & $p_{13}  = h(\compose(p_{12},p_{23})) = h(\esv \esc \esv \esd \esb \esd \esd) = \esv \esc \esv \esd \esd \esd$
    \end{tabular}
\end{subfigure}
}}
\caption{\scriptsize{Let $w_1 = abab$, $w_2=bcbbbab$, $w_3=ababab$.
\autoref{fig:p12} shows an optimal edit path $p_{12}$ between $w_1$ to $w_2$, \autoref{fig:p23} shows an optimal edit path $p_{23}$ between $w_2$ to $w_3$.
\autoref{fig:p13} shows the edit path $p_{13}$ composed from $p_{12}$ and $p_{23}$ using \autoref{def:compose}.
The edit operations in \autoref{fig:p13} are marked with numbers $1$ to $6$. 
A number $n$ in between $1$ and $6$ in \autoref{fig:p12} and \autoref{fig:p23} signifies that the corresponding
edge contributed to the construction of the edge marked $n$ in \autoref{fig:p13} (thus for the operations 
corresponding to cases (1) and (2) of \autoref{def:compose}, there is one corresponding marking in \autoref{fig:p12} and \autoref{fig:p23}
and for the others there are two). The labels ${-}$ in \autoref{fig:p12} and \autoref{fig:p23} correspond to case (11)
dealing with adding a letter when going from $s_1$ to $s_2$ and deleting it when going from $s_2$ to $s_3$, which yields the edit symbol \esb.
Note that $p_{13}$ is not optimal; still its cost is better than the sum of the costs of $p_{12}$ and $p_{23}$.}}\label{fig:compose}
\end{figure}

We further show that if the resulting string is $p_{13}$ then applying the function $\apply$ to $s_1$ and the edit path
obtained from $p_{13}$ by deleting all $\esb$ results in the string  $s_3$. 
\autoref{fig:compose} shows an example of the application of $\compose$ on two given edit paths.
In the sequel we will further show that the desired equations (\autoref{d}) and (\autoref{l}) hold.

Note that if we reach case (12) then we cannot claim that the result is an edit path.
We thus first show that if $\compose$ is applied to two edit paths $p_{12}, p_{23}$
such that $\apply(p_{12},s_1)=s_2$, and $\apply(p_{23},s_2)=s_3$, then 
the recursive application of $\compose(p_{12}, p_{23})$ will never reach the (12) case. That is,
$\compose(p_{12}, p_{23})$ is well defined.  

\begin{lemma}\label{lemma:compose_well_defined}
Let $s_1,s_2,s_3\in\Sigma^*$ and $p_{12},p_{23}\in\Gamma_\Sigma^*$ be edit paths, such that $\apply(p_{12},s_1)=s_2$ and $\apply(p_{23},s_2)=s_3$. 
	Then $p_{13}=\compose(p_{12},p_{23})$ is well-defined.
\end{lemma}

\begin{proof}
The proof is by structural induction  on $\compose$. 
For the \textbf{base case}, we have that $p_{12}=p_{23}=\varepsilon$.
Then $p_{13}=\varepsilon$. Thus $\compose$ reaches case (0) and is well defined.

For the \textbf{induction step} we have $p_{12}\neq \varepsilon$ or $p_{23}\neq \varepsilon$.
If $p_{12}=\varepsilon$ then it follows from the definition of $\apply$ that $s_1=s_2=\varepsilon$. 
Given that $\apply(p_{23},\varepsilon)$ is defined we get that $p_{23}[1] = \esv_\sigma$.
From the definition of $\apply$ we have $s_3=\sigma\cdot\apply(p_{23}[2..],s_2)$. Hence $s_3[2..]=\apply(p_{23}[2..],s_2)$.
Therefore, $\compose$ reaches case (2) and will never reach case (12) since
from the induction hypothesis it follows that $\compose(p_{12},p_{23}[2..])$ is well defined.

If $p_{23}=\varepsilon$ we get $s_2=s_3=\varepsilon$ and $p_{12}[1] = \esx_\sigma$.
Hence  $\compose$ reaches case (1) and similar reasoning shows that the induction hypothesis holds for the recursive application,
and thus the result is well defined.

Otherwise the first character of $p_{12}$ is not $\esx$ and the  first character of $p_{23}$ is not $\esv$.
We consider the remaining cases, by examining first the first letter of $p_{12}$.

\begin{enumerate}
    \item \textbf{Case} $p_{12}[1] = \esv_{\sigma_1}$.\\ 
    From the definition of $\apply$ we get that $s_2 = \sigma_1 \cdot s_2[2..]$ and $s_2[2..]=\apply(p_{12}[2..],s_1)$.
   
        \begin{enumerate}
            \item \textbf{Subcase} $p_{23}[1]=\esc_{(\sigma_2,\sigma_3)}$.\\
            From the definition of $\apply$ it follows that ${\sigma_1=\sigma_2}$, $s_3 = \sigma_3 \cdot s_3[2..]$ and $s_3[2..]=\apply(p_{23}[2..],s_2[2..])$.
            Thus $\compose$ reaches case (10) and the induction hypothesis holds for the recursive application.
            \item \textbf{Subcase} $p_{23}[1]=\esd_{\sigma_2}$.\\
            Similarly, from the definition of $\apply$ we get that $\sigma_1=\sigma_2$, $s_3 = \sigma_2 \cdot s_3[2..]$ and furthermore $s_3[2..]=\apply(p_{23}[2..],s_2[2..])$.
            Thus $\compose$ reaches case (9) and the induction hypothesis holds for the recursive application.
            \item \textbf{Subcase} $p_{23}[1]=\esx_{\sigma_2}$.\\
            Similarly, from the definition of $\apply$ we get that $\sigma_1=\sigma_2$ and $s_3=\apply(p_{23}[2..],s_2[2..])$.
            Thus $\compose$ reaches case (11) and the induction hypothesis holds for the recursive application.
        \end{enumerate}

    \item \textbf{Case} $p_{12}[1] = \esc_{(\sigma_1,\sigma_2)}$.\\ 
    From the definition of $\apply$ we get that $s_1 = \sigma_1 \cdot s_1[2..]$, $s_2 = \sigma_2 \cdot s_2[2..]$ and furthermore $s_2[2..]=\apply(p_{12}[2..],s_1[2..])$.
        
    \begin{enumerate}
        \item \textbf{Subcase} $p_{23}[1]=\esc_{(\sigma_3,\sigma_4)}$.\\
            From the definition of $\apply$ we get that ${\sigma_2=\sigma_3}$, $s_3 = \sigma_4 \cdot s_3[2..]$ and $s_3[2..]=\apply(p_{23}[2..],s_2[2..])$.
            Thus $\compose$ reaches case (6) and the induction hypothesis holds for the recursive application.
        \item \textbf{Subcase} $p_{23}[1]=\esd_{\sigma_3}$.\\
            Similarly, from the definition of $\apply$ it follows that ${\sigma_2=\sigma_3}$, $s_3 = \sigma_3 \cdot s_3[2..]$ and $s_3[2..]=\apply(p_{23}[2..],s_2[2..])$.
            Thus $\compose$ reaches case (8) and the induction hypothesis holds for the recursive application.
        \item \textbf{Subcase} $p_{23}[1]=\esx_{\sigma_3}$.\\
            Similarly, from the definition of $\apply$ we get that ${\sigma_2=\sigma_3}$ and $s_3=\apply(p_{23}[2..],s_2[2..])$.
            Thus $\compose$ reaches case (7) and the induction hypothesis holds for the recursive application.
    \end{enumerate}

    \item \textbf{Case} $p_{12}[1] = \esd_{\sigma}$\\
    From the definition of $\apply$ we get that $s_1 = \sigma\cdot s_1[2..]$, $s_2 = \sigma \cdot s_2[2..]$ and $s_2[2..]=\apply(p_{12}[2..],s_1[2..])$.
    
    \begin{enumerate}
        \item \textbf{Subcase} $p_{23}[1]=\esc_{(\sigma_1,\sigma_2)}$.\\
            From the definition of $\apply$ it follows that ${\sigma=\sigma_1}$, $s_3 = \sigma_2 \cdot s_3[2..]$ and $s_3[2..]=\apply(p_{23}[2..],s_2[2..])$.
            Thus $\compose$ reaches case (4) and the induction hypothesis holds for the recursive application.
        
        \item \textbf{Subcase} $p_{23}[1]=\esd_{\sigma_2}$.\\
            Similarly, from the definition of $\apply$ it follows that ${\sigma=\sigma_2}$, ${s_3 = \sigma_2 \cdot s_3[2..]}$ and furthermore $s_3[2..]=\apply(p_{23}[2..],s_2[2..])$.
            Thus $\compose$ reaches case (3) and the induction hypothesis holds for the recursive application.
        
        \item \textbf{Subcase} $p_{23}[1]=\esx_{\sigma_2}$.\\
            Similarly, from the definition of $\apply$ we get that ${\sigma=\sigma_2}$  and $s_3=\apply(p_{23}[2..],s_2[2..])$.
            Thus $\compose$ reaches case (5) and the induction hypothesis holds for the recursive application. \qedhere
    \end{enumerate}
\end{enumerate}
\end{proof}

Recall that the $\compose$ returns  a string over $\Gamma_\Sigma\cup\{\esb\}$
while $\apply$ first argument is expected to be a string over $\Gamma_\Sigma$.
We can convert the string  returned by $\compose$ to a string over $\Gamma_\Sigma$ 
 by simply removing the $\esb$ symbols.
To make this
precise we introduce 
 the function  $h\colon\Gamma_\Sigma \cup \{\esb\}\rightarrow \Gamma_\Sigma$  defined as follows
 $h(\gamma)= \varepsilon$ if $ \gamma=\esb $ and $h(\gamma)=\gamma$ otherwise;
 
and its natural extension $h\colon(\Gamma_\Sigma \cup \{\esb\})^*\rightarrow \Gamma_\Sigma^*$
defined as $h(\gamma_1\gamma_2\cdots\gamma_n)=h(\gamma_1)h(\gamma_2)\cdots h(\gamma_n)$.

We are now ready to state that $\compose$ fulfills its task, namely if it returns  $p_{13}$ then $h(p_{13})$ is an edit path from $s_1$ to $s_3$ and its weight and length satisfy \autoref{d} and \autoref{l}.
Note that  even if $p_{12}$ and $p_{23}$ are optimal, $h(p_{13})$ is not necessarily an optimal path from $s_1$ to $s_3$. Since the optimal path is no worse than $h(p_{13})$, it is enough for our purpose that $h(p_{13})$ is better than going through $s_2$.

\begin{proposition}\label{prop:properties-of-p13}
	Let $s_1,s_2,s_3\in\Sigma^*$ and $p_{12},p_{23}$ be edit paths, such that $\apply(p_{12},s_1)=s_2$, $\apply(p_{23},s_2)=s_3$. 
	Let $p_{13}=\compose(p_{12},p_{23})$. Let $d_{*}=\weight(p_{*})$ and $l_{*}=\len(p_{*})$ for $*\in\{12,23,12\}$.
	Then the following holds
	\vspace{-3mm}
	\begin{multicols}{3}
	\begin{enumerate}
		\item $\apply(h(p_{13}),s_1)=s_3$ 
		\item    $d_{13} \le d_{12}+d_{23}$
		\item $l_{13} \ge \max\{l_{12},l_{23}\}$
	\end{enumerate}	  
	\end{multicols}
\end{proposition}

\begin{proof}
The proof is by structural induction  on $\compose$. 
For the \textbf{base case}, we have that $p_{12}=p_{23}=\varepsilon$.
Then $p_{13}=\varepsilon$, by definition of apply we get that $s_1=s_2=s_3=\varepsilon$. Thus \\
1.  $\apply(h(p_{13}),s_1) = \apply(\varepsilon,\varepsilon) = \varepsilon=s_1=s_3$\\
2. and 3. we have that $d_{13} = 0 \leq d_{12}+d_{23}=0$ and $l_{13}= 0 \ge \max\{l_{12},l_{23}\} = 0$

For the \textbf{induction steps}, we have $p_{12}\neq \varepsilon$ or $p_{23}\neq \varepsilon$. 
Recall that $p_{13}=\compose(p_{12},p_{23})$. Thus, from \autoref{lemma:compose_well_defined} we can conclude $p_{13}$ is a string over $\Gamma_\Sigma\cup\{\esb\}$.
Let $s'_{*}=s_{*}[2..]$, $p'_{*}=p_{*}[2..]$, $d'_{*}=\weight(p'_{*})$, $l'_{*}=\len(p'_{*})$ for $*\in\{12,23,13\}$.
The proof  proceeds with the case analysis of $\compose$, going over  cases (1)-(11) of \autoref{def:compose}.

\begin{enumerate}[label=(\arabic*)]
    \item Here $p_{12}[1]=\esx_\sigma$. \\ 
    Then from $\apply$ we have $s_1 = \sigma\cdot s_1'$, from definition of $\compose$ we have $p_{13} = \esx_\sigma \cdot p_{13}'$\\
    Since $s_2=\apply(p_{12},s_1)=\apply(\esx_\sigma\cdot p_{12}',\sigma\cdot s_1')=\apply(p_{12}',s_1')$ and $\apply(p_{23},s_2)=s_3$, by applying the induction hypotheses on $s'_1,s_2,s_3$ we get
    \vspace{-3mm}
    \begin{itemize}[nosep]
        \item []
        \begin{multicols}{3}
        \begin{enumerate}[nosep,label=\arabic*.]
            \item $\apply(h(p_{13}'),s_1') = s_3$ 
            \item $d_{13}' \leq d_{12}' + d_{23}$ 
            \item $l_{13}' \ge \max\{l_{12}',l_{23}\}$
            \end{enumerate}
        \end{multicols}
    \end{itemize}
    \vspace{-2mm}
    Therefore
    \begin{itemize}[nosep]
        \item []
        \begin{enumerate}[nosep,label=\arabic*.]
            \item $\apply(h(p_{13}),s_1) = \apply(\esx_\sigma \cdot h(p_{13}'),\sigma\cdot s_1')= \apply(h(p_{13}'),s_1') = s_3$
            \item $d_{13} = 1 + d_{13}' \leq 1 + d_{12}' + d_{23} = d_{12} + d_{23}$
            \item $l_{13} = 1 + l_{13}' \ge 1 + \max\{l_{12}', l_{23}\} \ge \max\{1+l_{12}',l_{23}\} = \max\{l_{12},l_{23}\}.$
        \end{enumerate}
    \end{itemize}

    \item Here $p_{23}[1]=\esv_\sigma$.\\ 
    Then from $\apply$ we have $s_3 = \sigma\cdot s_3'$, from definition of $\compose$ we have $p_{13} = \esv_\sigma \cdot p_{13}'$.
    Since $\apply(p_{23},s_2)=\apply(\esv_\sigma\cdot p_{23}',s_2)=\sigma\cdot\apply(p_{23}',s_2)=s_3 = \sigma\cdot s_3'$ we get $\apply(p_{23}',s_2)=s_3'$ and $\apply(p_{12},s_1)=s_2$,
    by applying the induction hypotheses on $s_1,s_2,s_3'$ we get
    \vspace{-3mm}
    \begin{itemize}[nosep]
        \item []
        \begin{multicols}{3}
        \begin{enumerate}[nosep,label=\arabic*.]
            \item $\apply(h(p_{13}'),s_1) = s_3'$ 
            \item $d_{13}' \leq d_{12} + d_{23}'$ 
            \item $l_{13}' \ge \max\{l_{12},l_{23}'\}.$
        \end{enumerate}
        \end{multicols}
    \end{itemize}
    \vspace{-2mm}
    Therefore
    \begin{itemize}[nosep]
        \item []
        \begin{enumerate}[nosep,label=\arabic*.]
            \item $\apply(h(p_{13}),s_1) = \apply(\esv_\sigma \cdot h(p_{13}'),s_1)=
            \sigma\cdot \apply(h(p_{13}'),s_1) = \sigma\cdot s_3' = s_3$
            \item $d_{13} = 1 + d_{13}' \leq d_{12} + 1 + d_{23}' = d_{12} + d_{23}$
            \item $l_{13} = 1 + l_{13}' \ge 1 + \max\{l_{12}, l_{23}'\} \ge
            \max\{l_{12},1+l_{23}'\} = \max\{l_{12},l_{23}\}.$
       \end{enumerate}
    \end{itemize}

    \item Here $(p_{12}[1],p_{23}[1]) = (\esd_\sigma,\esd_\sigma)$.\\
    From the definition of $\compose$ we have $p_{13} = \esd_\sigma \cdot p_{13}'$ and from $\apply$ we have \\
    $\begin{array}{l@{\,=\,}l}\apply(p_{12},s_1) & \apply(\esd_\sigma\cdot p_{12}',\sigma\cdot s_1')=
    \sigma\cdot \apply(p_{12}',s_1') = \sigma\cdot s_2'=s_2\end{array}$ 
    and \\
    $\begin{array}{l}\apply(p_{23},s_2)=
    \apply(\esd_\sigma\cdot p_{23}',\sigma\cdot s_2')=
    \sigma\cdot \apply(p_{23}',s_2') = \sigma\cdot s_3'=s_3.\end{array}$\\
    Since $\apply(p_{12}',s_1') = s_2'$ and $\apply(p_{23}',s_2') = s_3'$
    , by applying the induction hypotheses on $s_1',s_2',s_3'$ we get
    \vspace{-3mm}
    \begin{itemize}[nosep]
        \item []
        \begin{multicols}{3}
        \begin{enumerate}[nosep,label=\arabic*.]
            \item $\apply(h(p_{13}'),s_1') = s_3'$ 
            \item $d_{13}' \leq d_{12}' + d_{23}'$ 
            \item $l_{13}' \ge \max\{l_{12}',l_{23}'\}.$
        \end{enumerate}
        \end{multicols}
    \end{itemize}
    \vspace{-2mm}
    Therefore
    \begin{itemize}[nosep]
        \item []
        \begin{enumerate}[nosep,label=\arabic*.]
            \item $\apply(h(p_{13}),s_1) = \apply(\esd_\sigma \cdot h(p_{13}'),\sigma\cdot s_1')= \sigma\cdot \apply(h(p_{13}'),s_1') = \sigma\cdot s_3' = s_3$
            \item $d_{13} = d_{13}' \leq d_{12}' + d_{23}' = d_{12} + d_{23}$
            \item $l_{13} = 1 + l_{13}' \ge 1 + \max\{l_{12}', l_{23}'\} =
            \max\{1+l_{12}',1+l_{23}'\} = \max\{l_{12},l_{23}\}.$
       \end{enumerate}
    \end{itemize}
    
    \item Here $(p_{12}[1],p_{23}[1]) = (\esd_{\sigma'},\esc_{(\sigma',\sigma)})$.\\
    By definition of compose we get $p_{13} = \esc_{(\sigma',\sigma)}\cdot p_{13}'$.
    From $\apply$ we have\\ 
    $\begin{array}{l}\apply(p_{12},s_1) = \apply(\esd_{\sigma'}\cdot p_{12}',\sigma'\cdot s_1')=
    \sigma'\cdot \apply(p_{12}',s_1') = \sigma'\cdot s_2'=s_2  \text{ and }\\
    \apply(p_{23},s_2) = \apply(\esc_{(\sigma',\sigma)}\cdot p_{23}',\sigma'\cdot s_2')= \sigma\cdot \apply(p_{23}',s_2') = \sigma\cdot s_3'=s_3.\end{array}$\\
    Since $\apply(p_{12}',s_1') = s_2'$ and $\apply(p_{23}',s_2') = s_3'$,
    by applying the induction hypotheses on $s_1',s_2',s_3'$ we get
    \vspace{-3mm}
    \begin{itemize}[nosep]
        \item []
        \begin{multicols}{3}
        \begin{enumerate}[nosep,label=\arabic*.]
            \item $\apply(h(p_{13}'),s_1') = s_3'$ 
            \item $d_{13}' \leq d_{12}' + d_{23}'$ 
            \item $l_{13}' \ge \max\{l_{12}',l_{23}'\}.$
        \end{enumerate}
        \end{multicols}
    \end{itemize}
    \vspace{-2mm}
    Therefore
    \begin{itemize}[nosep]
        \item []
        \begin{enumerate}[nosep,label=\arabic*.]
            \item 
            $\apply(h(p_{13}),s_1) = \apply(\esc_{(\sigma',\sigma)} \cdot h(p_{13}'),\sigma'\cdot s_1') = \sigma\cdot \apply(h(p_{13}'),s_1') = \sigma\cdot s_3' = s_3$
            \item $d_{13} = 1 + d_{13}' \leq d_{12}' + 1 + d_{23}' = d_{12} + d_{23}$
            \item $l_{13} = 1 + l_{13}' \ge 1 + \max\{l_{12}', l_{23}'\} = \max\{1+l_{12}',1+l'_{23}\} = \max\{l_{12},l_{23}\}.$
       \end{enumerate}
    \end{itemize}

    \item Here $(p_{12}[1],p_{23}[1]) = (\esd_\sigma,\esx_\sigma)$.\\
    By definition of compose we get that $p_{13} = \esx_\sigma\cdot p_{13}'$.
    From $\apply$ we have\\
     $\begin{array}{l}\apply(p_{12},s_1) = \apply(\esd_\sigma\cdot p_{12}',\sigma\cdot s_1')= \sigma\cdot \apply(p_{12}',s_1') = \sigma\cdot s_2'=s_2  \text{ and }\\
    \apply(p_{23},s_2) = \apply(\esx_\sigma\cdot p_{23}',\sigma\cdot s_2') = \apply(p_{23}',s_2') = s_3.\end{array}$\\
    Since $\apply(p_{12}',s_1') = s_2'$ and $\apply(p_{23}',s_2') = s_3$
    , by applying the induction hypotheses on $s_1',s_2',s_3$ we get
    \vspace{-3mm}
    \begin{itemize}[nosep]
        \item []
        \begin{multicols}{3}
        \begin{enumerate}[nosep,label=\arabic*.]
            \item $\apply(h(p_{13}'),s_1') = s_3$ 
            \item $d_{13}' \leq d_{12}' + d_{23}'$ 
            \item $l_{13}' \ge \max\{l_{12}',l_{23}'\}.$
        \end{enumerate}
        \end{multicols}
    \end{itemize}
    \vspace{-2mm}
    Therefore
    \begin{itemize}[nosep]
        \item []
        \begin{enumerate}[nosep,label=\arabic*.]
            \item $\apply(h(p_{13}),s_1) = \apply(\esx_\sigma \cdot h(p_{13}'),\sigma\cdot s_1')= \apply(h(p_{13}'),s_1') = s_3$
            \item $d_{13} = 1 + d_{13}' \leq d_{12}' + 1 + d_{23}' = d_{12} + d_{23}$
            \item $l_{13} = 1 + l_{13}' \ge 1 + \max\{l_{12}', l_{23}'\} =
            \max\{1+l_{12}',1+l'_{23}\} = \max\{l_{12},l_{23}\}.$
       \end{enumerate}
    \end{itemize}

    \item Here $(p_{12}[1],p_{23}[1]) = (\esc_{(\sigma_1,\sigma_2)},\esc_{(\sigma_2,\sigma_3)})$.\\
    By definition of compose we get that $p_{13} = \esc_{(\sigma_1,\sigma_3)}\cdot p_{13}'$.
    From $\apply$ we have\\
    $\begin{array}{l}\apply(p_{12},s_1) =   \apply(\esc_{(\sigma_1,\sigma_2)}\cdot p_{12}',\sigma_1\cdot s_1') = \sigma_2\cdot \apply(p_{12}',s_1') = \sigma_2\cdot s_2'=s_2  \text{ and }\\
    \apply(p_{23},s_2) = \apply(\esc_{(\sigma_2,\sigma_3)}\cdot p_{23}',\sigma_2\cdot s_2') = \sigma_3\cdot \apply(p_{23}',s_2') = \sigma_3\cdot s_3'=s_3.\end{array}$\\
    Since $\apply(p_{12}',s_1') = s_2'$ and $\apply(p_{23}',s_2') = s_3'$
    , by applying the induction hypotheses on $s_1',s_2',s_3'$ we get
    \vspace{-3mm}
    \begin{itemize}[nosep]
        \item []
        \begin{multicols}{3}
        \begin{enumerate}[nosep,label=\arabic*.]
            \item $\apply(h(p_{13}'),s_1') = s_3'$ 
            \item $d_{13}' \leq d_{12}' + d_{23}'$ 
            \item $l_{13}' \ge \max\{l_{12}',l_{23}'\}.$
        \end{enumerate}
        \end{multicols}
    \end{itemize}
    \vspace{-2mm}
    Therefore
    \begin{itemize}[nosep]
        \item []
        \begin{enumerate}[nosep,label=\arabic*.]
           \item 
           $\apply(h(p_{13}),s_1) = \apply(\esc_{(\sigma_1,\sigma_3)} \cdot h(p_{13}'),\sigma_1\cdot s_1') = \sigma_3\cdot \apply(h(p_{13}'),s_1') = \sigma_3 s_3' = s_3$
            \item $d_{13} = 1 + d_{13}' \leq 1 + d_{12}' + d_{23}' < 1 + d_{12}' + 1 + d_{23}' = d_{12} + d_{23}$
            \item $l_{13} = 1 + l_{13}' \ge 1 + \max\{l_{12}', l_{23}'\} = \max\{1+l_{12}',1+l'_{23}\} = \max\{l_{12},l_{23}\}.$
       \end{enumerate}
    \end{itemize}

    \item Here $(p_{12}[1],p_{23}[1]) = (\esc_{(\sigma_1,\sigma_2)},\esx_{\sigma_2})$.\\
    By definition of compose we get that $p_{13} = \esx_{\sigma_1}\cdot p_{13}'$.
    From $\apply$ we have\\
    $\begin{array}{l}\apply(p_{12},s_1) = \apply(\esc_{(\sigma_1,\sigma_2)}\cdot p_{12}',\sigma_1\cdot s_1')=
    \sigma_2\cdot \apply(p_{12}',s_1') = \sigma_2\cdot s_2'=s_2  \text{ and } \\
    \apply(p_{23},s_2) = \apply(\esx_{\sigma_2}\cdot p_{23}',\sigma_2\cdot s_2')= \apply(p_{23}',s_2') = s_3.\end{array}$\\
    Since $\apply(p_{12}',s_1') = s_2'$ and $\apply(p_{23}',s_2') = s_3$
    , by applying the induction hypotheses on $s_1',s_2',s_3$ we get
    \vspace{-3mm}
    \begin{itemize}[nosep]
        \item []
        \begin{multicols}{3}
        \begin{enumerate}[nosep,label=\arabic*.]
            \item $\apply(h(p_{13}'),s_1') = s_3$ 
            \item $d_{13}' \leq d_{12}' + d_{23}'$ 
            \item $l_{13}' \ge \max\{l_{12}',l_{23}'\}.$
        \end{enumerate}
        \end{multicols}
    \end{itemize}
    \vspace{-2mm}
    Therefore
    \begin{itemize}[nosep]
        \item []
        \begin{enumerate}[nosep,label=\arabic*.]
           \item $\apply(h(p_{13}),s_1) = \apply(\esx_{\sigma_1} \cdot h(p_{13}'),\sigma_1\cdot s_1')= \apply(h(p_{13}'),s_1') = s_3$
            \item $d_{13} = 1 + d_{13}' \leq 1 + d_{12}' + d_{23}' < 1 + d_{12}' + 1 + d_{23}' = d_{12} + d_{23}$
            \item $l_{13} = 1 + l_{13}' \ge 1 + \max\{l_{12}', l_{23}'\} =
            \max\{1+l_{12}',1+l'_{23}\} = \max\{l_{12},l_{23}\}.$
       \end{enumerate}
    \end{itemize}

    \item Here $(p_{12}[1],p_{23}[1]) = (\esc_{(\sigma',\sigma)},\esd_\sigma)$.\\
    By definition of compose we get that $p_{13} = \esc_{(\sigma',\sigma)}\cdot p_{13}'$.
    From $\apply$ we have\\
    $\begin{array}{l}\apply(p_{12},s_1) = \apply(\esc_{(\sigma',\sigma)}\cdot p_{12}',\sigma'\cdot s_1') = \sigma\cdot \apply(p_{12}',s_1') = \sigma\cdot s_2'=s_2  \text{ and }\\
    \apply(p_{23},s_2) = \apply(\esd_\sigma\cdot p_{23}',\sigma\cdot s_2') = \sigma\cdot \apply(p_{23}',s_2') = \sigma\cdot s_3'=s_3.\end{array}$\\
    Since $\apply(p_{12}',s_1') = s_2'$ and $\apply(p_{23}',s_2') = s_3'$,
    by applying the induction hypotheses on $s_1',s_2',s_3'$ we get
    \vspace{-3mm}
    \begin{itemize}[nosep]
        \item []
        \begin{multicols}{3}
        \begin{enumerate}[nosep,label=\arabic*.]
            \item $\apply(h(p_{13}'),s_1') = s_3'$ 
            \item $d_{13}' \leq d_{12}' + d_{23}'$ 
            \item $l_{13}' \ge \max\{l_{12}',l_{23}'\}.$
        \end{enumerate}
        \end{multicols}
    \end{itemize}
    \vspace{-2mm}
    Therefore
    \begin{itemize}[nosep]
        \item []
        \begin{enumerate}[nosep,label=\arabic*.]
           \item 
           $\apply(h(p_{13}),s_1) = \apply(\esc_{(\sigma',\sigma)} \cdot h(p_{13}'),\sigma'\cdot s_1')= 
           \sigma\cdot \apply(h(p_{13}'),s_1') = \sigma\cdot s_3' = s_3$
            \item $d_{13} = 1 + d_{13}' \leq 1 + d_{12}' + d_{23}' = d_{12} + d_{23}$
            \item $l_{13} = 1 + l_{13}' \ge 1 + \max\{l_{12}', l'_{23}\} =
            \max\{1+l_{12}',1+l'_{23}\} = \max\{l_{12},l_{23}\}.$
       \end{enumerate}
    \end{itemize}

    \item Here $(p_{12}[1],p_{23}[1]) = (\esv_\sigma,\esd_\sigma)$.\\
    By definition of compose we get that $p_{13} = \esv_\sigma\cdot p_{13}'$.
    From $\apply$ we have\\
    $\begin{array}{l}\apply(p_{12},s_1) = \apply(\esv_\sigma\cdot p_{12}',s_1) = \sigma\cdot \apply(p_{12}',s_1) = \sigma\cdot s_2'=s_2  \text{ and }\\
    \apply(p_{23},s_2) = \apply(\esd_\sigma\cdot p_{23}',\sigma\cdot s_2') = \sigma\cdot \apply(p_{23}',s_2') = \sigma\cdot s_3'=s_3.\end{array}$\\
    Since $\apply(p_{12}',s_1) = s_2'$ and $\apply(p_{23}',s_2') = s_3'$,
    by applying the induction hypotheses on $s_1,s_2',s_3'$ we get
    \vspace{-3mm}
    \begin{itemize}[nosep]
        \item []
        \begin{multicols}{3}
        \begin{enumerate}[nosep,label=\arabic*.]
            \item $\apply(h(p_{13}'),s_1) = s_3'$ 
            \item $d_{13}' \leq d_{12}' + d_{23}'$ 
            \item $l_{13}' \ge \max\{l_{12}',l_{23}'\}.$
        \end{enumerate}
        \end{multicols}
    \end{itemize}
    \vspace{-2mm}
    Therefore
    \begin{itemize}[nosep]
        \item []
        \begin{enumerate}[nosep,label=\arabic*.]
           \item $\apply(h(p_{13}),s_1) = \apply(\esv_\sigma \cdot h(p_{13}'),s_1)=
            \sigma\cdot \apply(h(p_{13}'),s_1) = \sigma\cdot s_3' = s_3$
            \item $d_{13} = 1 + d_{13}' \leq 1 + d_{12}' + d_{23}' = d_{12} + d_{23}$
            \item $l_{13} = 1 + l_{13}' \ge 1 + \max\{l_{12}', l_{23}'\} =
            \max\{1+l_{12}',1+l'_{23}\} = \max\{l_{12},l_{23}\}.$
       \end{enumerate}
    \end{itemize}

    \item Here $(p_{12}[1],p_{23}[1]) = (\esv_{\sigma_1},\esc_{(\sigma_1,\sigma_2)})$.\\
    By definition of compose we get that $p_{13} = \esv_{\sigma_2}\cdot p_{13}'$.
    From $\apply$ we have\\
    $\begin{array}{l}\apply(p_{12},s_1) = \apply(\esv_{\sigma_1}\cdot p_{12}',s_1) = \sigma_1\cdot \apply(p_{12}',s_1) = \sigma_1\cdot s_2'=s_2  \text{ and }\\
    \apply(p_{23},s_2) = \apply(\esc_{(\sigma_1,\sigma_2)}\cdot p_{23}',\sigma_1\cdot s_2') = \sigma_2\cdot \apply(p_{23}',s_2') = \sigma_2\cdot s_3'=s_3.\end{array}$\\
    Since $\apply(p_{12}',s_1) = s_2'$ and $\apply(p_{23}',s_2') = s_3'$,
    by applying the induction hypotheses on $s_1,s_2',s_3'$ we get
    \vspace{-3mm}
    \begin{itemize}[nosep]
        \item []
        \begin{multicols}{3}
        \begin{enumerate}[nosep,label=\arabic*.]
            \item $\apply(h(p_{13}'),s_1) = s_3'$ 
            \item $d_{13}' \leq d_{12}' + d_{23}'$ 
            \item $l_{13}' \ge \max\{l_{12}',l_{23}'\}.$
        \end{enumerate}
        \end{multicols}
    \end{itemize}
    \vspace{-2mm}
    Therefore
    \begin{itemize}[nosep]
        \item []
        \begin{enumerate}[nosep,label=\arabic*.]
            \item 
            $\apply(h(p_{13}),s_1) = \apply(\esv_{\sigma_2} \cdot h(p_{13}'),s_1)= \sigma_2\cdot \apply(h(p_{13}'),s_1) = \sigma_2\cdot s_3' = s_3$
            \item $d_{13} = 1 + d_{13}' \leq 1 + d_{12}' + d_{23}' <$ $ 1 + d_{12}' + 1 + d_{23}' = d_{12} + d_{23}$
            \item $l_{13} = 1 + l_{13}' \ge 1 + \max\{l_{12}', l_{23}'\} = \max\{1+l_{12}',1+l'_{23}\} = \max\{l_{12},l_{23}\}.$
       \end{enumerate}
    \end{itemize}

    \item Here $(p_{12}[1],p_{23}[1]) = (\esv_\sigma,\esx_\sigma)$.\\
    By definition of compose we get that $p_{13} = \esb\cdot p_{13}'$.
    From $\apply$ we have\\
    $\begin{array}{l}\apply(p_{12},s_1) = \apply(\esv_\sigma\cdot p_{12}',s_1) = \sigma\cdot \apply(p_{12}',s_1) = \sigma\cdot s_2'=s_2  \text{ and }\\
    \apply(p_{23},s_2) = \apply(\esx_\sigma\cdot p_{23}',\sigma\cdot s_2') = \apply(p_{23}',s_2')=s_3.\end{array}$\\
    Since $\apply(p_{12}',s_1) = s_2'$ and $\apply(p_{23}',s_2') = s_3$,
     by applying the induction hypotheses on $s_1,s_2',s_3$ we get
    \vspace{-3mm}
    \begin{itemize}[nosep]
        \item []
        \begin{multicols}{3}
        \begin{enumerate}[nosep,label=\arabic*.]
            \item $\apply(h(p_{13}'),s_1) = s_3$ 
            \item $d_{13}' \leq d_{12}' + d_{23}'$ 
            \item $l_{13}' \ge \max\{l_{12}',l_{23}'\}.$
        \end{enumerate}
        \end{multicols}
    \end{itemize}
    \vspace{-2mm}
    Therefore
    \begin{itemize}[nosep]
        \item []
        \begin{enumerate}[nosep,label=\arabic*.]
            \item $\apply(h(p_{13}),s_1) = \apply(h(p_{13}'),s_1) = s_3$
            \item $d_{13} = 2 + d_{13}' \leq 1 + d_{12}' + 1 + d_{23}' = d_{12} + d_{23}$
            \item $l_{13} = 2 + l_{13}' > 1 + \max\{l_{12}', l_{23}'\} =
            \max\{1+l_{12}',1+l'_{23}\} = \max\{l_{12},l_{23}\}.$ \qedhere
       \end{enumerate}
    \end{itemize}
\end{enumerate}
\end{proof}

The sequel makes use of the following lemmas regarding non-negative integers $d$ and $l$.
\begin{lemma}\label{plus1}
	 If $d\le l$ then $\frac{d+1}{l+1} \ge \frac{d}{l}$
\end{lemma}

\begin{proof}
	$\frac{d+1}{l+1} = \frac{l(d+1)}{l(l+1)} \ge
	\frac{d(l+1)}{l(l+1)}= \frac{d}{l}.$
\end{proof}

\begin{lemma}\label{plus}
	 If $d_{13} \le d_{12}+d_{23}$ and
	$l_{13} \ge \max\{l_{12},l_{23}\}$ 
	then $\frac{d_{12}}{l_{12}}+\frac{d_{23}}{l_{23}}
	\ge \frac{d_{13}}{l_{13}}$.
\end{lemma}

\begin{proof} $\frac{d_{13}}{l_{13}} \le \frac{d_{12}+d_{23}}{l_{13}} =
	\frac{d_{12}}{l_{13}} + \frac{d_{23}}{l_{13}} \le
	\frac{d_{12}}{l_{12}} + \frac{d_{23}}{l_{23}}$.
\end{proof}

Recall that $\cost$ is defined as $\weight$ divided by $\len$.
Let $p_{13}$ be the string obtained by compose in \autoref{prop:properties-of-p13}.
Then by items 2 and 3 we know that
\vspace{-10mm}
\begin{multicols}{2}
\begin{equation}
  \weight(p_{13}) \le \weight(p_{12})+\weight(p_{23}) 
\end{equation}
\begin{equation}
  \len(p_{13}) \ge \max\{\len(p_{12}),\len(p_{23})\}
\end{equation}
\end{multicols}
\noindent
We can thus conclude from \autoref{plus} that the cost of the path obtained by $\compose$
is at most the sum of the costs of the edit paths from which
it was obtained, as stated in the following corollary.

\begin{corollary}\label{cor:cost-without-h}
	Let $s_1,s_2,s_3\in\Sigma^*$ and $p_{12},p_{23}$ be edit paths, such that $\apply(p_{12},s_1)=s_2$, $\apply(p_{23},s_2)=s_3$. 
	Let $p_{13}=\compose(p_{12},p_{23})$. Then $\cost(p_{13})\leq \cost(p_{12})+\cost(p_{23})$.
\end{corollary}

We are not done yet, since $p_{13}$ contains $\esb$ symbols, and thus it is not really an edit path.
Let $k$ be the number of $\esb$'s in $p_{13}$. Then $\weight(p_{13})=2 k + \weight(h(p_{13}))$
and $\len(p_{13})=2 k + \len(h(p_{13}))$, applying $2k$ times~\autoref{plus1}, we conclude that
$\begin{array}{l}\frac{\weight(p_{13})}{\len(p_{13})}\ge\frac{\weight(h(p_{13}))}{\len(h(p_{13}))}.\end{array}$

\begin{corollary}\label{cor:removing-h}
$\cost(p)\ge\cost(h(p))$
\end{corollary}

\begin{proposition}\label{prop:triangle}
    The normalized edit distance obeys the triangle inequality.
\end{proposition}
\begin{proof}
Let $s_1,s_2,s_3\in\Sigma^*$ and $p_{12},p_{23}$ be  optimal edit paths. 
That is, $\apply(p_{12},s_1)=s_2$ and $\apply(p_{23},s_2)=s_3$ and 
$\ned(s_1,s_2)= \cost(p_{12})$ and $\ned(s_2,s_3)= \cost(p_{23})$.
	Let $p_{13}=\compose(p_{12},p_{23})$.
	From \autoref{cor:cost-without-h} we get that $\cost(p_{13})\leq \cost(p_{12})+\cost(p_{23})$.
	From \autoref{prop:properties-of-p13} it holds that $h(p_{13})$ is a valid edit path over $\Gamma_\Sigma$.
	From \autoref{cor:removing-h} we get that $\cost(h(p_{13}))\leq\cost(p_{13})$.
	By definition of \ned\ as it chooses the minimal cost of an edit path, $\ned(s_1,s_3)\leq \cost(h(p_{13}))$.
	To conclude, we get
	$\ned(s_1,s_3) \leq \ned(s_1,s_2)+\ned(s_2,s_3)$. 
\end{proof}

\begin{theorem}
The Normalized Levenshtein Distance $\ned$ (provided in \autoref{def:ned}) 
with uniform costs (i.e., where the cost of all inserts, deletes and swaps are some constant $c$) is a metric
on the space $\Sigma^*$.
\end{theorem}
\begin{proof}
    The first two conditions of being a metric follow from \autoref{prop:reflexivity-and-symmetry}. The third condition, namely triangle inequality, follows from \autoref{prop:triangle}.
\end{proof}

%% file: tikz/x_y.tex
\begin{tikzpicture}
\draw[step=1cm, color=gray] (0, 0) grid (6,4);
\draw[->,line width=3pt, rounded corners] (0, 4) -- (1, 3) -- (2, 3) -- (3, 2) -- (4, 2) -- (5, 1) -- (6, 0);


\node[circle,fill=red!20, inner sep=0pt, minimum size=15pt] (n1) at (0.5,3.5) {2};
\node[circle,fill=red!20, inner sep=0pt, minimum size=15pt] (n2) at (1.5,3) {3};
\node[circle,fill=red!20, inner sep=0pt, minimum size=15pt] (n3) at (2.5,2.5) {4};
\node[circle,fill=red!20, inner sep=0pt, minimum size=15pt] (n4) at (3.5,2) {-};
\node[circle,fill=red!20, inner sep=0pt, minimum size=15pt] (n5) at (4.5,1.5) {5};
\node[circle,fill=red!20, inner sep=0pt, minimum size=15pt] (n6) at (5.5,0.5) {6};

\draw[] (0.5,3.5)++(-90:.3) node[below]{\\ \textcolor{blue}{\esc}};
\draw[] (1.5,3)++(-90:.3) node[below]{\\ \textcolor{blue}{\esv}};
\draw[] (2.5,2.5)++(-90:.3) node[below]{\\ \textcolor{blue}{\esd}};
\draw[] (3.5,2)++(-90:.3) node[below]{\\ \textcolor{blue}{\esv}};
\draw[] (4.5,1.5)++(-90:.3) node[below]{\\ \textcolor{blue}{\esd}};
\draw[] (5.5,0.5)++(-90:.3) node[below]{\\ \textcolor{blue}{\esd}};

\foreach \coord/\label [count=\xi] in {{0,1}/{$b$},
{0,2}/{$a$},
{0,3}/{$b$},
{0,4}/{$a$}}{
    \node[anchor=north east] at (\coord) {\label};
}
\foreach \coord/\label [count=\xi] in {{0,4}/{$b$},
{1,4}/{$c$},
{2,4}/{$b$},
{3,4}/{$b$},
{4,4}/{$a$},
{5,4}/{$b$}}{
    \node[anchor=south west] at (\coord) {\label};
}
\end{tikzpicture}
  

%% file: tikz/x_z_cmps.tex
\begin{tikzpicture}
\draw[step=1cm, color=gray] (0, 0) grid (6,4);

\draw[->,line width=3pt, rounded corners] (0, 4) -- (1, 4) -- (2, 3) -- (3, 3) -- (4, 2) -- (5, 1) -- (6, 0);


\node[circle,fill=red!20, inner sep=0pt, minimum size=15pt] (n1) at (0.5,4) {1};
\node[circle,fill=red!20, inner sep=0pt, minimum size=15pt] (n2) at (1.5,3.5) {2};
\node[circle,fill=red!20, inner sep=0pt, minimum size=15pt] (n3) at (2.5,3) {3};
\node[circle,fill=red!20, inner sep=0pt, minimum size=15pt] (n4) at (3.5,2.5) {4};
\node[circle,fill=red!20, inner sep=0pt, minimum size=15pt] (n5) at (4.5,1.5) {5};
\node[circle,fill=red!20, inner sep=0pt, minimum size=15pt] (n6) at (5.5,0.5) {6};

\draw[] (0.5,4)++(-90:.3) node[below]{\\ \textcolor{blue}{\esv}};
\draw[] (1.5,3.5)++(-90:.3) node[below]{\\ \textcolor{blue}{\esc}};
\draw[] (2.5,3)++(-90:.3) node[below]{\\ \textcolor{blue}{\esv}};
\draw[] (3.5,2.5)++(-90:.3) node[below]{\\ \textcolor{blue}{\esd}};
\draw[] (4.5,1.5)++(-90:.3) node[below]{\\ \textcolor{blue}{\esd}};
\draw[] (5.5,0.5)++(-90:.3) node[below]{\\ \textcolor{blue}{\esd}};

\foreach \coord/\label [count=\xi] in {{0,1}/{$b$},
{0,2}/{$a$},
{0,3}/{$b$},
{0,4}/{$a$}}{
    \node[anchor=north east] at (\coord) {\label};
}
\foreach \coord/\label [count=\xi] in {{0,4}/{$a$},
{1,4}/{$b$},
{2,4}/{$a$},
{3,4}/{$b$},
{4,4}/{$a$},
{5,4}/{$b$}}{
    \node[anchor=south west] at (\coord) {\label};
}
\end{tikzpicture}
  

%% file: tikz/y_z.tex
\begin{tikzpicture}
\draw[step=1cm, color=gray] (0, 0) grid (6,6);
\draw[->,line width=3pt, rounded corners] (0, 6) -- (1, 6) -- (2, 5) -- (3, 4) -- (4, 3) -- (4, 2) -- (5, 1) -- (6, 0);


\node[circle,fill=red!20, inner sep=0pt, minimum size=15pt] (n1) at (0.5,6) {1};
\node[circle,fill=red!20, inner sep=0pt, minimum size=15pt] (n2) at (1.5,5.5) {2};
\node[circle,fill=red!20, inner sep=0pt, minimum size=15pt] (n3) at (2.5,4.5) {3};
\node[circle,fill=red!20, inner sep=0pt, minimum size=15pt] (n4) at (3.5,3.5) {4};
\node[circle,fill=red!20, inner sep=0pt, minimum size=15pt] (n5) at (4,2.5) {-};
\node[circle,fill=red!20, inner sep=0pt, minimum size=15pt] (n5) at (4.5,1.5) {5};
\node[circle,fill=red!20, inner sep=0pt, minimum size=15pt] (n6) at (5.5,0.5) {6};

\draw[] (0.5,6)++(-90:.3) node[below]{\\ \textcolor{blue}{\esv}};
\draw[] (1.5,5.5)++(-90:.3) node[below]{\\ \textcolor{blue}{\esd}};
\draw[] (2.5,4.5)++(-90:.3) node[below]{\\ \textcolor{blue}{\esc}};
\draw[] (3.5,3.5)++(-90:.3) node[below]{\\ \textcolor{blue}{\esd}};
\draw[] (3.5,2.2)++(90:0.3) node[below]{\\ \textcolor{blue}{\esx}};
\draw[] (4.5,1.5)++(-90:.3) node[below]{\\ \textcolor{blue}{\esd}};
\draw[] (5.5,0.5)++(-90:.3) node[below]{\\ \textcolor{blue}{\esd}};

\foreach \coord/\label [count=\xi] in {{0,1}/{$b$},
{0,2}/{$a$},
{0,3}/{$b$},
{0,4}/{$b$},
{0,5}/{$c$},
{0,6}/{$b$}}{
    \node[anchor=north east] at (\coord) {\label};
}
\foreach \coord/\label [count=\xi] in {{0,6}/{$a$},
{1,6}/{$b$},
{2,6}/{$a$},
{3,6}/{$b$},
{4,6}/{$a$},
{5,6}/{$b$}}{
    \node[anchor=south west] at (\coord) {\label};
}
\end{tikzpicture}
  